\documentclass[11pt]{amsart}

\usepackage{amsmath,amssymb,amsthm}
\usepackage{a4wide}
\usepackage{algorithm}
\usepackage{algorithmic}

\usepackage{graphicx}

\newtheorem{lemma}{Lemma}[section]
\newtheorem{theorem}[lemma]{Theorem}

\newtheorem{definition}[lemma]{Definition}

\newtheorem{conjecture}[lemma]{Conjecture}

\newcommand{\nices}{\mathcal{S}}
\newcommand{\instanceS}{(S;S_1,\ldots,S_m)}
\newcommand{\edge}{-}

\newcommand{\dist}{\textrm{dist}}
\newcommand{\core}{\textrm{core}}
\newcommand{\rooot}{\textrm{link}}
\newcommand{\link}{\rooot}
\newcommand{\depth}{\textrm{depth}}
\newcommand{\nicep}{\mathcal{P}}
\newcommand{\girth}[1]{\mathcal{GIRTH}_{#1}}
\newcommand{\girthplus}[1]{\mathcal{GIRTH}^+_{#1}}
\newcommand{\hc}{H2C}

\begin{document}


\title{Large-girth roots of graphs}

\author[Anna Adamaszek, Micha{\l} Adamaszek]{Anna Adamaszek$^{1,3}$, Micha{\l} Adamaszek$^{2,3}$}
\thanks{\textit{Subject classification}: Algorithms and data structures}
\thanks{Research supported by the Centre for Discrete
        Mathematics and its Applications (DIMAP), EPSRC award EP/D063191/1.}
\thanks{$^1$Department of Computer Science, University of Warwick, Coventry, CV4 7AL, UK}
\thanks{$^2$Warwick Mathematics Institute, University of Warwick, Coventry, CV4 7AL, UK}
\thanks{$^3$Centre for Discrete Mathematics and its Applications (DIMAP), University of Warwick}
\thanks{E-mails: \texttt{\{annan,aszek\}@mimuw.edu.pl}}



\begin{abstract}
We study the problem of recognizing graph powers and computing roots of graphs. We provide a polynomial time recognition algorithm for $r$-th powers of graphs of girth at least $2r+3$, thus improving a bound  conjectured by Farzad et al. (STACS 2009). Our algorithm also finds all $r$-th roots of a given graph that have girth at least $2r+3$ and no degree one vertices, which is a step towards a recent conjecture of Levenshtein that such root should be unique. On the negative side, we prove that recognition becomes an NP-complete problem when the bound on girth is about twice smaller. Similar results have so far only been attempted for $r=2,3$.
\end{abstract}


\maketitle

\section{Introduction}
\label{section:introduction}

All graphs in this paper are simple, undirected and connected. If $H$ is a graph, its \emph{$r$-th power} $G=H^r$ is the graph on the same vertex set such that two distinct vertices are adjacent in $G$ if their distance in $H$ is at most $r$. We also call $H$ the \emph{$r$-th root} of $G$.

There are some problems naturally related to graph powers and graph roots. Suppose $\nicep$ is a class of graphs (possibly consisting of all graphs), $r$ is an integer and $G$ is an arbitrary graph. The questions we ask are:
\begin{itemize}
\item \textit{The recognition problem}: Is $G$ an $r$-th power of some graph from $\nicep$? Formally, we define a family of decision problems:

\begin{tabular}{lll}
  & \textbf{Problem.} & $r$-TH-POWER-OF-$\nicep$-GRAPH\\
  & \textbf{Instance.} & A graph $G$.\\
  & \textbf{Question.} & Is $G=H^r$ for some graph $H\in\nicep$?
\end{tabular}
\item \textit{The $r$-th root problem}: Find some/all $r$-th roots of $G$ which belong to $\nicep$.
\item \textit{The unique reconstruction problem}: Is the $r$-th root of $G$ in $\nicep$ (if any) unique? 
\end{itemize}

The above problems have been investigated for various graph classes $\nicep$. There exist characterizations of squares \cite{Muk} and higher powers \cite{EMR} of graphs, but they are not computationally efficient. Motwani and Sudan \cite{MotSud} proved the NP-completeness of recognizing graph squares and Lau \cite{Lau} extended this to cubes of graphs. Motwani and Sudan \cite{MotSud} suggested that recognizing squares of bipartite graphs is also likely to be NP-complete. This was disproved by Lau \cite{Lau}, who gave a polynomial time algorithm that recognizes squares of bipartite graphs and counts the bipartite square roots of a given graph. Apparently the first proof that $r$-TH-POWER-OF-GRAPH and $r$-TH-POWER-OF-BIPARTITE-GRAPH are NP-complete for any $r\geq 3$ was recently announced in \cite{LeNg}.

Considerable attention has been given to tree roots of graphs, which are quite well understood and can be computed efficiently. Lin and Skiena \cite{Lin} gave a polynomial time algorithm for recognizing squares of trees. The first polynomial time algorithm for $r$-TH-POWER-OF-TREE for arbitrary $r$ was given by Kearney and Corneil \cite{KeaCor}. A faster, linear time algorithm for this problem is due to Chang, Ko and Lu \cite{CKL}. All these algorithms also compute some $r$-th tree root of a given graph (if one exists). It is important to note that such a root need not be unique, not even up to isomorphism, so the difficulty lies in making consistent choices while constructing a root. We are going to use the computation of an $r$-th tree root of a graph as a black-box in our algorithms.

There has also been some work on the complexity of $r$-TH-POWER-OF-$\nicep$-GRAPH for such classes $\nicep$ as chordal graphs, split graphs and proper interval graphs \cite{LauCor} and for directed graphs and their powers \cite{Kutz}.

In this work we address the above problems for another large family of graphs, namely graphs with no short cycles. Recall, that the \emph{girth} of a graph is the length of its shortest cycle. For convenience we shall denote by  $\girth{\geq g}$ the class of all graphs of girth at least $g$, and by $\girthplus{\geq g}$ its subclass consisting of graphs with no vertices of degree one (which we call \emph{leaves}). These classes of graphs make a convenient setting for graph roots because of the possible uniqueness results outlined below.

By \cite{4auth} the recognition of squares of $\girth{\geq 4}$-graphs is NP-complete, while squares of $\girth{\geq 6}$-graphs can be recognized in polynomial time. For $r\geq 3$ no complexity-theoretic results have been known, but there is some very interesting work on the uniqueness of the roots. Precisely, Levenshtein et al. \cite{Lev+} proved that if $G$ has a square root $H$ in the class  $\girthplus{\geq 7}$, then $H$ is unique\footnote{It is not possible to obtain uniqueness if the vertices of degree one are allowed, hence this technical restriction. See \cite{Lev+} for details.}. The same statement was extended in \cite{Lev} to $r$-th roots in $\girthplus{\geq 2r+2\lceil(r-1)/{4}\rceil+1}$. The main conjecture in this area remains unresolved:
\begin{conjecture}[Levenshtein, \cite{Lev}]
\label{conjecture:zajebiste}
 If a graph $G$ has an $r$-th root $H$ in $\girthplus{\geq 2r+3}$, then $H$ is unique in that class.
\end{conjecture}
The value of $g=2r+3$ is best possible, as witnessed by the cycle $C_{2r+2}$, which cannot be uniquely reconstructed from its $r$-th power.
The best result towards Conjecture \ref{conjecture:zajebiste} is that the number of roots $H$ under consideration is at most $\Delta(G)$ (the maximum vertex degree in $G$, \cite{Lev}), but its proof yields only exponential time $r$-th root and recognition algorithms.

At the same time Farzad et al. made a conjecture about recognizing powers of graphs of lower-bounded girth:
\begin{conjecture}[Farzad et al., \cite{4auth}]
\label{conjecture:debilne}
 The problem $r$-TH-POWER-OF-$\girth{\geq 3r-1}$-GRAPH can be solved in polynomial time.
\end{conjecture}

\paragraph{\textbf{Our contribution.}} Our first result gives an efficient reconstruction algorithm in Levenshtein's case:
\begin{theorem}
\label{theorem:1}
Given any graph $G$, all its $r$-th roots in $\girthplus{\geq 2r+3}$ can be found in polynomial time.
\end{theorem}

Next, we use this result to deal with the general case, i.e. when the roots are allowed to have leaves. It turns out that the same girth bound of $2r+3$ admits a positive result:

\begin{theorem}
\label{theorem:positive}
The problem $r$-TH-POWER-OF-$\girth{\geq 2r+3}$-GRAPH can be solved in polynomial time.
\end{theorem}

Our result proves Conjecture \ref{conjecture:debilne} (for $r\geq 4$) and is in fact stronger. It also improves the result of \cite{LeNg} for $r=3,g=10$. Moreover, our algorithm for this problem is constructive and exhaustive in the sense that it finds ``all'' $r$-th roots in $\girth{\geq 2r+3}$ modulo the non-uniqueness of $r$-th tree roots of graphs, as explained in Section \ref{section:with-leafs}.

These positive results have a hardness counterpart:
\begin{theorem}
\label{theorem:girth}
The problem $r$-TH-POWER-OF-$\girth{\geq g}$-GRAPH is NP-complete for $g\leq r+1$ when $r$ is odd and $g\leq r+2$ when $r$ is even.
\end{theorem}


The paper is structured as follows. First we prove some auxiliary results, useful both in the construction of algorithms and in the hardness result. Section \ref{section:algorithm} contains the main algorithm from Theorem \ref{theorem:1}, which is then used in Section \ref{section:with-leafs} as a building block of the general recognition algorithm from Theorem \ref{theorem:positive}. NP-completeness is proved in Section \ref{section:large-girth}. 

\section{Auxiliary results}
\label{section:reductions}

Let us fix some terminology. By $\dist_H(u,v)$ we denote the distance from $u$ to $v$ in $H$. The \emph{$d$-neighbourhood} of a vertex $u$ in $H$ is the set of vertices of $H$ which are exactly in distance $d$ from $u$. The $1$-neighbourhood (i.e. the set of vertices adjacent to $u$) will be denoted $N_H(u)$.

Our setup usually involves a pair of graphs $G$ and $H$ on a common vertex set $V$ such that $G=H^r$. We adopt the notation
$$B_v:=\{u\in V: \dist_H(u,v)\leq r\}=N_G(v)\cup\{v\}$$
for $v\in V$ (the letter $B$ stands for ``ball'' of radius $r$ in $H$). The lack of explicit reference to $r$ and $H$ in this notation should not lead to confusion. It is important that $B_v$ depend only on $G$.

Almost all previous work on algorithmic aspects of graph powers \cite{MotSud,4auth,Lau,LauCor,LeNg} makes use of a special gadget, called \emph{tail structure}, which, applied to a vertex $u$ in $G$, ensures that in any $r$-th root $H$ of $G$ this vertex has the same, pre-determined neighbourhood. Our main observation is that in fact such a tail structure carries a lot more information about $H$. It pins down not just $N_H(u)$, but also each $d$-neighbourhood of $u$ in $H$ for $d=1,\ldots,r$.

\begin{lemma}
\label{lemma:gadget}
Let $G=H^r$ and suppose that $\{v_0,v_1,\ldots,v_r\}\subset V$ is a set of vertices such that $N_G(v_r)=\{v_{r-1},\ldots,v_1,v_0\}$ and $N_G(v_{i+1})\subset N_G(v_{i})$ for all $i=0,\ldots,r-1$, where the inclusions are strict.
\footnote{This assumption (strictness of inclusions) can be removed at the cost of a more complicated statement, but this generality is not needed here.}

Then the subgraph of $H$ induced by $\{v_0,v_1,\ldots,v_r\}$ is a path $v_0\edge v_1\edge\ldots\edge v_r$ and the $d$-neighbourhood of $v_0$ in $H$ is precisely
$$N_G(v_{r-d})\setminus N_G(v_{r-d+1})\cup\{v_d\}$$
for all $d=1,\ldots,r$.
\end{lemma}
\begin{proof}
The subgraph $K$ of $H$ induced by $\{v_0,\ldots,v_r\}$ is connected --- otherwise $N_G(v_r)$ would contain vertices from outside $K$. Consider any vertex $u$ of $K$ that has an edge to some vertex $w$ outside $K$. Clearly, $\dist_K(v_r,u)=r$, since otherwise $w$ would be in $N_G(v_r)$. This means that $K$ is a path from $v_r$ to $u$ and $u$ is the only vertex of that path which has edges to vertices outside $K$. The condition $N_G(v_{i+1})\subset N_G(v_{i})$ now implies that the vertices of this path are arranged as in the conclusion of the lemma. The second conclusion follows easily.
\end{proof}
Note that the tail structure itself does not enforce any extra constraints on $H$ other than the $d$-neighbourhoods of $v_0$.

In the algorithm for $r$-TH-POWER-OF-$\girth{\geq 2r+3}$-GRAPH we will need to solve the following tree root problem with additional restrictions imposed on the $d$-neighbourhoods of a certain vertex:\\

\begin{tabular}{lll}
  & \textbf{Problem.} & RESTRICTED-$r$-TH-TREE-ROOT\\
  & \textbf{Instance.} & A graph $G$, $r\geq 2$, a vertex $v\in V(G)$ and a partition\\
  & & $V(G)=\{v\}\cup T^{(1)}\cup\ldots\cup T^{(r)}\cup T^{(>r)}$.\\
  & \textbf{Question.} & Is $G=T^r$ for some tree $T$ such that the\\
  & & $d$-neighbourhood of $v$ in $T$ is exactly $T^{(d)}$ for $d=1,\ldots,r$?
\end{tabular}

\begin{lemma}
There is a constructive polynomial time algorithm for RESTRICTED-$r$-TH-TREE-ROOT.
\end{lemma}
\begin{proof}
Define an auxiliary graph $G'$ by 
\begin{align*}
V(G')= V(G) & \cup\{w_1,\ldots,w_r\}\cup\{u_1,\ldots,u_r\}\\
E(G')= E(G) & \cup\{w_i\edge w_j, u_i\edge u_j \textrm{ for all } i,j=1,\ldots,r\}\\
&  \cup\{w_i\edge u_j \textrm{ if } i+j\leq r\}\\
&  \cup\{w_i\edge v, u_i\edge v \textrm{ for all } i=1,\ldots,r\}\\
& \cup\{w_i\edge x, u_i\edge x \textrm{ for all } x\in T^{(j)} \textrm{ if } i+j\leq r\}
\end{align*}

We claim that the instance of RESTRICTED-$r$-TH-TREE-ROOT has a solution if and only if $G'$ has an $r$-th tree root (with no restrictions). Indeed, if our instance is solvable, then the solution can be turned into an $r$-th root of $G'$ by appending two paths $v\edge w_1\edge\ldots\edge w_r$ and $v\edge u_1\edge\ldots\edge u_r$ at $v$. On the other hand both sets $\{v,w_1,\ldots,w_r\}$ and $\{v,u_1,\ldots,u_r\}$ satisfy the assumptions of Lemma \ref{lemma:gadget}
\footnote{We used two paths just to ensure that the inclusions in Lemma \ref{lemma:gadget} are strict regardless of how small the rest of the graph might be. Again, with a more complicated statement of that lemma one path would suffice.}
, which implies that any $r$-th root of $G'$ has those two paths as induced subgraphs and that the $d$-neighbourhood of $v$ is $T^{(d)}\cup\{w_d,u_d\}$ for $d=1,\ldots,r$. It means that a solution to the instance of RESTRICTED-$r$-TH-TREE-ROOT can be obtained by searching for any $r$-th tree root of $G'$ and omitting the vertices $w_i, u_i$. For this we can use the algorithms of \cite{KeaCor,CKL}.
\end{proof}

\section{Algorithm for roots in $\girthplus{\geq 2r+3}$}
\label{section:algorithm}
In this section we present the algorithm from Theorem \ref{theorem:1}, that is the polynomial time reconstruction of all $r$-th roots in $\girthplus{\geq 2r+3}$ of a given graph $G$. There are two structural properties of graphs $H\in\girthplus{\geq 2r+3}$ that will be used freely throughout the proofs:
\begin{itemize}
\item Every $x\in V(H)$ is of degree at least 2 and the subgraph of $H$ induced by $B_x$ is a tree. This holds since any cycle in $H$ within $B_x$ would have length at most $2r+1$. We shall depict the ball $B_x$ in $H$ in the tree-like fashion.
\item If there is a simple path from $u$ to $v$ in $H$ of length exactly $r+1$ or $r+2$ then $u\not\in B_v$. Indeed, $u\in B_v$ iff there is a path of length at most $r$ from $u$ to $v$ in $H$, and combined with the first path this would yield a cycle of length at most $2r+2$.
\end{itemize}

To describe the algorithm we introduce the following sets defined for every $x,y \in V$.
\begin{align*}
S_{x,y}&=B_x\cap B_y\setminus\bigcup_{v\in B_y\setminus B_x} B_v \setminus \{x\}\\
P_{x,y}&=B_x\cap B_y\cap\bigcup_{v\in S_{x,y}} B_v\\
N_{x,y}&=B_x\cap B_y\cap\bigcap_{v\in P_{x,y}} B_v \setminus \{x\}
\end{align*}

These sets become meaningful if we compute them for the endpoints of an actual edge in some $r$-th root of $G$. Precisely:

\begin{theorem}
\label{theorem:edgedetermines}
Suppose $G=H^r$ for a graph $H\in\girthplus{\geq 2r+3}$ and $xy\in E(H)$. Then
$$N_{x,y}=N_H(x).$$
\end{theorem}
\begin{proof}
Because of the girth condition the set $B_x\cup B_y$ in $H$ consists of two disjoint trees $T_x$ and $T_y$, rooted in $x$ and $y$ respectively and connected by the edge $xy$ (see Fig.\ref{fig:2}). Let us introduce some subsets of those trees. By $W_x$ and $W_y$ denote the last levels:
$$W_x=\{u\in T_x: \dist_H(u,x)=r\}, \quad W_y=\{u\in T_y: \dist_H(u,y)=r\},$$
by $P_x$ and $P_y$ the next-to-last levels:
$$P_x=\{u\in T_x: \dist_H(u,x)=r-1\}, \quad P_y=\{u\in T_y: \dist_H(u,y)=r-1\},$$
and by $N_x$ and $N_y$ the children of $x$ and $y$ in $T_x$ and $T_y$:
$$N_x=\{u\in T_x: \dist_H(u,x)=1\},\quad N_y=\{u\in T_y: \dist_H(u,y)=1\}.$$
Clearly $B_x\cap B_y=(T_x\setminus W_x)\cup (T_y\setminus W_y)$, $W_x=B_x\setminus B_y$ and $W_y=B_y\setminus B_x$. Note that if $r=2$ we have $N_x=P_x$ and $N_y=P_y$.

\begin{figure}[h!]
\includegraphics{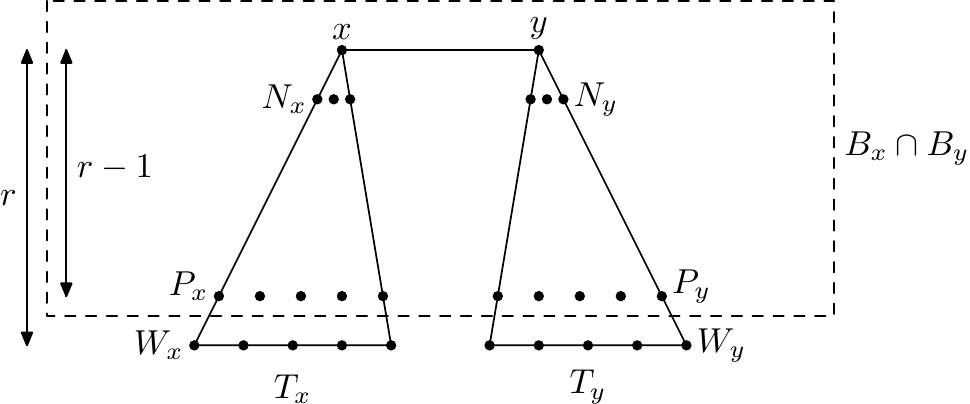}
\caption{The subgraph of $H$ induced by $B_x\cap B_y$.}
\label{fig:2}
\end{figure}

First observe that every $u\in N_x$ and every $v\in B_y\setminus B_x=W_y$ are connected by a path of length $r+2$. It follows that $u\not\in B_v$, which implies
$$N_x\subset S_{x,y}.$$

It is also clear that $S_{x,y}\subset T_x$ (because every vertex in $T_y$ has a descendant $v\in W_y$).

Now the sum $\bigcup_{v\in S_{x,y}} B_x\cap B_y\cap B_v$ contains $\bigcup_{v\in N_x}B_x\cap B_y\cap B_v = (B_x\cap B_y)\setminus P_y$. On the other hand, if $v\in S_{x,y}$ and $u\in P_y$ then $u\not\in B_v$. Indeed, if $u\in B_v$ then there would be a path from $u$ to $v$ of length at most $r$. This path cannot be contained in $T_x\cup T_y$ (because $\dist_H(u,x)=r$, so one can only get as far as $x$ going from $u$), hence it must exit $T_y$ through $W_y$ and then enter $T_x$ through $W_x$, finally reaching $v\in S_{x,y}$. However, that yields a path from $W_y$ to $S_{x,y}$ of length at most $r$ (in fact at most $r-1$), contradicting the definition of $S_{x,y}$. Eventually we proved
$$P_{x,y}=(B_x\cap B_y)\setminus P_y.$$

Now we have $\{y\}\cup N_x\subset N_{x,y}$ because every vertex of $\{y\}\cup N_x$ is in distance at most $r$ from all the vertices of $(B_x\cap B_y)\setminus P_y$. On the other hand, for every vertex $u$ of $B_x\cap B_y$ that is not in $N_x\cup\{x,y\}$ one can find a path of length $r+1$ that starts in $u$ and ends in a vertex $v\in (B_x\cap B_y)\setminus P_y$. Then $u\not\in B_v$, so $u\not\in N_{x,y}$. Such a path is obtained by going from $u$ up the tree it is contained in ($T_x$ or $T_y$) and then down in the other tree.

Concluding, we have identified $N_{x,y}$ to be $N_x\cup\{y\}$, as required.
\end{proof}

The previous theorem should be understood as follows. Given a graph $G$, we want to find its $r$-th root $H$. If we fix at least one edge $xy$ of $H$ in advance, we can compute the neighbourhood $N_H(x)$ of $x$ using only the data available in $G$. But then we can move on in the same way, computing the neighbours of those neighbours etc.

\begin{algorithm}[h!]
\caption{\textbf{Input:} $G$,$r$. \textbf{Output:} All $r$-th roots of $G$ in $\girthplus{\geq 2r+3}$}
\begin{algorithmic}
\STATE pick a vertex $x$ with smallest $|B_x|$
\FOR{all $y$ in $B_x$}
 \STATE $H=$reconstructFromOneEdge($G,xy$)
 \STATE \textbf{if} $H\in\girthplus{\geq 2r+3}$ and $H^r=G$ \textbf{output} $H$
\ENDFOR
\end{algorithmic}

\begin{algorithmic}
\STATE reconstructFromOneEdge($G,e$):
 \STATE $H=(V(G),\{e\})$
 \STATE \textbf{for} all $u\in V$ \textbf{set} processed[$u$]$:=$\textbf{false}
 \WHILE{$H$ has an unprocessed vertex $x$ of degree at least $1$}
   \STATE $y$ = any neighbour of $x$ in $H$
   \STATE $E(H)=E(H)\cup\{xz \textrm{ for all } z\in N_{x,y}\}$
   \STATE processed[$x$]$:=$\textbf{true}
 \ENDWHILE
 \STATE \textbf{return} $H$
\end{algorithmic}
\end{algorithm}

The $r$-th root algorithm is now straightforward. The procedure \emph{reconstructFromOneEdge} attempts to compute $H$ from $G$ assuming the existence of a given edge $e$ in $H$. This is repeated for all possible edges from a fixed vertex $x$. If $e$ is an edge in some $r$-th root $H$ of $G$, then Theorem \ref{theorem:edgedetermines} guarantees that we will recover exactly $H$ (in fact many times, once for each edge $xy\in H$; we omit the obvious optimization which avoids this redundancy).
~\\
\paragraph{\textbf{Remark.}} With an appropriate list representation of $G$ the set $N_{x,y}$ can be determined in time $O(|E(G)|+|V(G)|)$ for any $x,y$, so the running time of \emph{reconstructFromOneEdge} is $O(|V(G)|\cdot|E(G)|)$. By choosing the initial $x$ to be the vertex of the smallest degree in $G$ we can achieve the total running time of $O(\frac{|E(G)|}{|V(G)|}\cdot |V(G)|\cdot |E(G)|)=O(|E(G)|^2)$.

\section{Removing the no-leaves restriction}
\label{section:with-leafs}

In this section we obtain a polynomial time algorithm for the general recognition problem $r$-TH-POWER-OF-$\girth{\geq 2r+3}$-GRAPH, proving Theorem \ref{theorem:positive}. We start with a few definitions (see Fig.\ref{fig:8}).

For a graph $H$, which is not a tree, let $\core(H)$ denote the largest subgraph of $H$ with no vertices of degree one (leaves). Alternatively this can be defined as follows. Given $H$, let $H'$ be the graph obtained from $H$ by removing all leaves and inductively define $H^{(n)} = (H^{(n-1)})'$. This process eventually stabilizes at the graph $\core(H)$.

A vertex $v\in V(H)$ is called a \emph{core vertex} if it belongs to $\core(H)$ and a \emph{non-core vertex} otherwise. The non-core vertices are grouped into trees attached to the core. For every vertex $v\in\core(H)$ we denote by $T_v$ the tree attached at $v$ (including $v$) and by $T_v^{(d)}$ (for $d\geq 0$) the set of vertices of $T_v$ located in distance $d$ from $v$. For a non-core vertex $u$ the \emph{link} of $u$ (denoted $\link(u)$) is its closest core vertex and the \emph{depth} of $u$ (denoted $\depth(u)$) is the distance from $u$ to $\link(u)$. 

\begin{figure}
\includegraphics[scale=1]{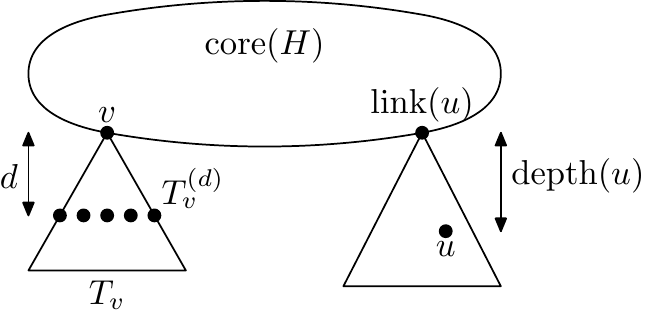}
\caption{The notation of Section \ref{section:with-leafs}.}
\label{fig:8}
\end{figure}

\subsection{Outline of the algorithm.} The algorithm for $r$-TH-POWER-OF-$\girth{\geq 2r+3}$-GRAPH processes the input graph $G$ in several steps (see Algorithm 2). First, we check if $G$ has a tree $r$-th root \cite{KeaCor,CKL}. If not, then we split the vertices of $G$ into the core and non-core vertices of any of its $r$-th roots. Lemma \ref{lemma:leaves} ensures that this partition is uniquely determined only by the graph $G$. 

Let $\tilde{G}$ be the subgraph of $G$ induced by all the vertices that are classified as belonging to the core of any possible $r$-th root $H$. We now employ the algorithm from the previous section to find all $r$-th roots $\tilde{H}$ of $\tilde{G}$ which have girth at least $2r+3$ and no leaves (there is $O(\Delta(G))$ of them; conjecturally there is at most one).

Finally, we must attach the non-core vertices to each of the possible $\tilde{H}$. It turns out that once the core is fixed, the link of each non-core vertex can be uniquely determined, so we can pin down all the sets $V(T_v)$. However, we cannot simply look for any $r$-th tree root of the subgraph of $G$ induced by $V(T_v)$, because we have to ensure that the tree structure that we are going to impose on $V(T_v)$ is compatible with the neighbourhood information contained in the rest of $G$. Fortunately Lemma \ref{lemma:depths} guarantees that for a fixed $G$ and $\core(H)$, all the sets $T_v^{(d)}$ for $d=1,\ldots,r$ are also uniquely determined. Since all the distances from the vertices of $T_v$ to the rest of the graph depend only on the vertex depths and the structure of the core, this is exactly the additional piece of data we need. Any tree root satisfying the given depth constraints will be compatible with the rest of the graph. Concluding, the problem we are left with for each $T_v$ is the RESTRICTED-$r$-TH-TREE-ROOT from Section \ref{section:reductions}. If all these instances have positive solutions, then the graph $H$ defined as $\tilde{H}$ with the trees $T_v$ attached at each core vertex $v$ is an $r$-th root of $G$. 

The next two subsections describe the two crucial steps: detecting non-core vertices and the reconstruction of trees $T_v$.

\begin{algorithm}[h!]
\label{algorithm:2}
\caption{\newline \textbf{Input:} $G$,$r$.\newline \textbf{Output:} $r$-th roots of $G$ in $\girth{\geq 2r+3}$ (one per each core)}
\begin{algorithmic}
\STATE check if $G=T^r$ for some tree $T$
\STATE
\STATE $\tilde{G}:=G$
\WHILE{$\tilde{G}$ has vertices $u,v$ with $B_u\subset B_v$}
 \STATE remove from $\tilde{G}$ all $u$ such that $B_u\subset B_v$ for some $v$
\ENDWHILE
\STATE
\FOR{every graph $\tilde{H}\in\girthplus{\geq 2r+3}$ such that $\tilde{H}^r=\tilde{G}$}
 \STATE $H:=\tilde{H}$
 \FOR{every vertex $v\in V(\tilde{H})$}
    \STATE find $V(T_v)$ and a partition $V(T_v)=\{v\}\cup T_v^{(1)}\cup\ldots\cup T_v^{(r)}\cup T_v^{(>r)}$
    \STATE use $restrictedTreeRoot$ to reconstruct some tree $T_v$
    \STATE extend $H$ by attaching $T_v$ at $v$
 \ENDFOR
 \STATE \textbf{if} all $T_v$ existed \textbf{output} $H$
\ENDFOR
\end{algorithmic}
\end{algorithm}

\subsection{Finding core and non-core vertices.}

The next lemma shows how to detect all vertices located ``close to the bottom'' of the trees $T_v$ in $H$.

\begin{lemma}
\label{lemma:leaves}
Suppose $H\in\girth{\geq 2r+3}$ and $H^r=G$.Then the following conditions are equivalent for a vertex $u\in H$:
\begin{itemize}
\item[(1)] There is some other vertex $v\in H$ such that $B_u\subset B_v$.
\item[(2)] $u\not\in H^{(r)}$.
\end{itemize}
\end{lemma}
\begin{proof}
If $u\in H^{(r)}$ then $u$ is not removed in the first $r$ steps of cutting off the leaves of $H$, which means there exist at least two disjoint paths of length $r$ starting at $u$. However, it implies that for every vertex $v\in B_u$ there exists another $v'\in B_u$ (on one of those paths) such that $\dist_H(v,v')=r+1$, hence $v'\in B_u$ but $v'\not\in B_v$. Therefore $B_u$ is not contained in $B_v$ for any $v\neq u$.

If, on the other hand, $u\not\in H^{(r)}$, then $u$ becomes a leaf after at most $r-1$ steps of the leaf-removal procedure and is removed in the subsequent step. Let $v$ be the last vertex adjacent to $u$ just before $u$ is removed. Clearly $B_u\subset B_v$.
\end{proof}

An inductive repetition of the above criterion determines the consecutive sets $V(H^{(r)})$, $V(H^{(2r)})$, $V(H^{(3r)})$, $\ldots$ for \emph{any} $r$-th root $H\in\girth{\geq 2r+3}$ of $G$ using only the information available in $G$. Eventually we obtain $V(\core(H))$ which is the vertex set of $\tilde{G}$.

\textbf{Remark.} Repeated application of Lemma \ref{lemma:leaves} also proves that if $G$ has a tree $r$-th root then it does not have a non-tree $r$-th root in $\girth{\geq 2r+3}$ and vice-versa.

\subsection{Attaching the trees $T_v$.}
For each possible $\core(H)$ we need to decide on a way of attaching the remaining (non-core) vertices to $H$ in a way which ensures that $H^r=G$. It turns out that all the data necessary to ensure the compatibility can be read off from $G$ and $\core(H)$, so again this data is common for all the possible $r$-th roots of $G$ that have a fixed core. 

\begin{lemma}
\label{lemma:depths}
Suppose that $H\in\girth{\geq 2r+3}$ is a graph such that $H$ is not a tree and $H^r=G$. Then for every non-core vertex $u$ of $H$ we have: 
\begin{itemize}
\item either $B_u\cap V(\core(H))=\emptyset$, in which case $\depth(u)>r$, or
\item the subgraph of $H$ induced by $B_u\cap V(\core(H))$ is a tree whose only center is $\link(u)$ and whose height (the distance from the center to every leaf) is $r-\depth(u)$.
\end{itemize}
\end{lemma}
\begin{proof}
The first statement is obvious. As for the second, the subgraph induced by $B_u\cap V(\core(H))$ consists of all the vertices of $V(\core(H))$ in distance at most $r-\depth(u)$ from $\link(u)$. Since $\core(H)$ is a graph of girth at least $2r+3$ with no degree one nodes, these vertices induce a tree in $H$, and all the leaves of this tree are exactly in distance $r-\depth(u)$ from $\link(u)$. Therefore $\link(u)$ is the unique center of that tree.
\end{proof}

Lemma \ref{lemma:depths} yields a method of partitioning the non-core vertices into the sets $V(T_v)$ and subdividing each $V(T_v)$ into a disjoint union $\{v\}\cup T_v^{(1)}\cup\ldots\cup T_v^{(r)}\cup T_v^{(>r)}$ of vertices in distance $1,2,\ldots,r$ and more than $r$ from $v$ using only the data from $G$ and $\core(H)$. Indeed, for the vertices $u$ with $B_u\cap V(\core(H))\not=\emptyset$ one finds the center and height of the subtree of $\core(H)$ induced by $B_u\cap V(\core(H))$ and applies the second part of Lemma \ref{lemma:depths} to obtain both $\link(u)$ and $\depth(u)$, thus classifying $u$ to the appropriate $T_v^{(d)}$. The links of all remaining vertices are determined using the fact that all vertices in one connected component of $G\setminus\bigcup_{v\in \core(H), d=0,\ldots,r-1} T_v^{(d)}$ have the same link.

\textbf{Remark.} The above partition can also be obtained (perhaps in a computationally easier way) from the following fact:

\begin{lemma}
If $H\in\girth{\geq 2r+3}$ is a graph such that $H$ is not a tree and $H^r=G$, then for every vertex $v$ of $\core(H)$ and every $d=1,\ldots,r$ we have:
$$T_v^{(d)} = \bigcap_{\substack{a\in \core(H)\\ \dist_H(v,a)\leq r-d}}B_a\setminus \bigcup_{\substack{b\in \core(H)\\ \dist_H(v,b)\geq r-d+1}}B_b.$$
\end{lemma}
\begin{proof}

The $\subset$ inclusion is obvious. Now suppose $u$ belongs to the right-hand side. If $u$ is in $T_v$, then it clearly must have depth $d$, so it suffices to prove that $u$ cannot belong to any other $T_{v'}$. Suppose this is the case: $\rooot(u)=v'\neq v$ and $l:=\depth(u)$, $1\leq l\leq r$. Set $k=\dist_H(v,v')$ and let $\Gamma$ denote the shortest path from $v$ to $v'$. 

If $r-d\leq k-1$ then let $a,b\in\Gamma$ be the vertices in distance $r-d$ and $r-d+1$ from $v$, respectively. By assumption $u\in B_a$ and $u\not\in B_b$, but that is impossible since $b$ is closer to $u$ than $a$.

When $r-d\geq k$ extend the path $\Gamma$ beyond $v'$ in $\core(H)$, so as to reach two vertices $a$ and $b$ in distance $r-d$ and $r-d+1$ from $v$, respectively (this is possible in $\core(H)$). The assumption $u\in B_a\setminus B_b$ enforces $\dist_H(u,a)=r$, but $\dist_H(u,a)=l+(r-d-k)$ so $l=k+d$.

Now extend the path $\Gamma$ beyond $v$ up to a point $w$ such that $\dist_H(v,w)=r-k-l+1$ (this is possible because $u\in B_v$ implies $k+l\leq r$ and because we are in the core). Moreover, the path $u\edge v'\edge v\edge w$ has length $r+1$ so it measures the distance from $u$ to $w$ and proves that $u\not\in B_w$. On the other hand:
$$\dist_H(v,w)=r-k-l+1=r-k-(k+d)+1=r-d+1-2k<r-d$$
so we ought to have $u\in B_w$. This contradiction ends the proof.
\end{proof}

\section{Hardness results}
\label{section:large-girth}

Now we proceed to the hardness of recognition for powers of graphs of lower-bounded girth (Theorem \ref{theorem:girth}). For the reductions we use the following NP-complete problem (see \cite[Prob. SP4]{Gar}). It has already been successfully applied in this context (\cite{4auth, Lau, LauCor, LeNg}).

\begin{tabular}{llp{12cm}}
  & \textbf{Problem.} & HYPERGRAPH 2-COLORABILITY  (\hc{}) \\
  & \textbf{Instance.} & A finite set $S$ and a collection $S_1,\ldots,S_m$ of subsets of $S$. \\
  & \textbf{Question.} & Can the elements of $S$ be colored with two colors $A$, $B$ such that each set $S_j$ has elements of both colors?
\end{tabular}

An instance of this problem (also known as SET-SPLITTING) will be denoted $\nices=\instanceS$. We shall refer to the elements of the universum $S$ as $x_1,\ldots, x_n$. Any assignment of colors $A$ and $B$ to the elements of $S$ which satisfies the requirements of the problem will be called a \emph{2-coloring}.

In this section we fix $r$ and let $k=\lfloor \frac{r}{2}\rfloor$, so that $r=2k$ or $r=2k+1$ depending on parity. We shall define a graph that encodes both the structure of the \hc{} instance and the coloring. To ensure large girth, the connections between vertices representing sets, elements and colors will be realized by paths of length $k$. Since the graph under consideration is rather large we shall describe its succinct representation that does not require the enumeration of all edges. 

\subsection{Case of odd $r=2k+1$}
Consider an instance $\nices=\instanceS$ of \hc. The following two definitions describe an auxiliary graph that will be used as a base for further constructions. The reader is referred to Fig.\ref{fig:red-odd} for a self-explanatory presentation of the graphs $K_\nices$ and $H_\nices$ defined below.

\begin{definition}
For an instance $\nices=\instanceS$ let $V_\nices$ be the following set of vertices:
\begin{itemize}
\item $S_j, x_i$ for all subsets and elements,
\item $A, B, X$, 
\item $T_{i,j}^{(l)}$ for every pair $i,j$ such that $x_i\in S_j$ and every $l=1,\ldots,k-1$,
\item $P_{i}^{(l)}$ for every $x_i$ and every $l=1,\ldots,k-1$,
\item the tail vertices $S_j^{(l)}$ for each $j$ and $l=1,\ldots,r$.
\end{itemize}
\end{definition}

\begin{definition}
Given any instance $\nices=\instanceS$ define a graph $K_\nices$ on the vertex set $V_\nices$ with the following edges:
\begin{itemize}
\item a path $S_j\edge T_{i,j}^{(1)}\edge\ldots\edge T_{i,j}^{(k-1)}\edge x_i$ whenever $x_i\in S_j$,
\item a path $x_i\edge P_{i}^{(1)}\edge\ldots\edge P_{i}^{(k-1)}$ for every $x_i$,
\item $X\edge x_i$ for all $i$,
\item the tail paths, that is $S_j\edge S_j^{(1)}\edge S_j^{(2)}\edge\ldots\edge S_j^{(r)}$ for every $j$.
\end{itemize}
\end{definition}

\begin{figure}
\includegraphics[scale=0.9]{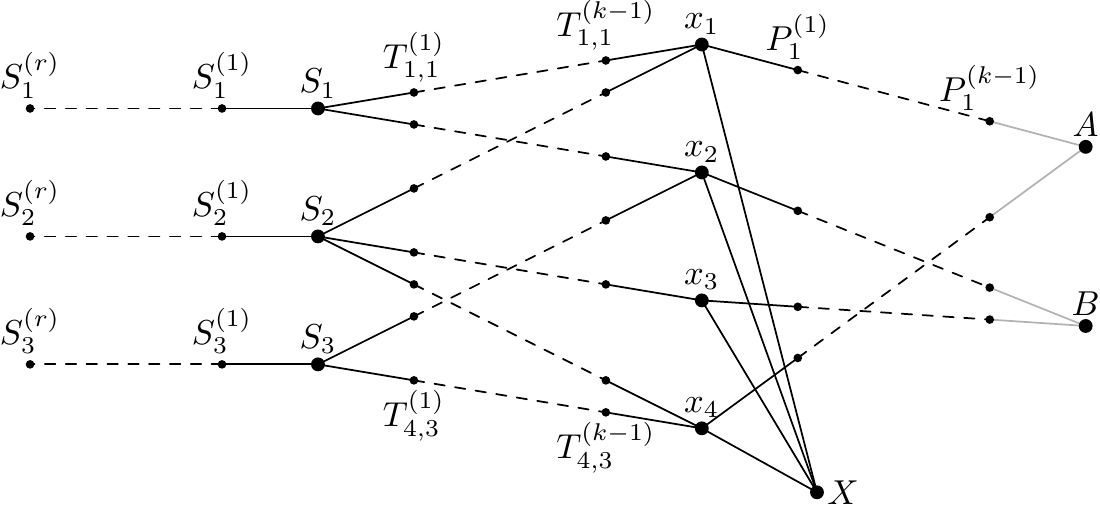}
\caption{For $\nices=(\{x_1,\ldots,x_4\}; \{x_1,x_2\},\{x_1,x_3,x_4\},\{x_2,x_4\}\})$ the graph $K_\nices$ consists of all \emph{but} the shaded edges. The graph $H_\nices$ (made of all the edges above) encodes the coloring with $x_1,x_4$ of color $A$ and $x_2,x_3$ of color $B$. It is a 2-coloring of $\nices$ since all $S_j$ are in distance $2k$ from $A$ and $B$.}
\label{fig:red-odd}
\end{figure}

This graph encodes only the structure of $\nices$. To encode the coloring we link the loose paths from $x_i$ to either $A$ or $B$.

\begin{definition}
Given an instance $\nices$ and a color assignment, define the graph $H_\nices$ to be $K_\nices$ with the additional edges $P_i^{(k-1)}\edge A$ whenever $x_i$ has color $A$ and $P_i^{(k-1)}\edge B$ whenever $x_i$ has color $B$.
\end{definition}

Note that $H_\nices$ has girth $2k+2=r+1$. Now comes the graph to be used in our NP-completeness reduction:

\begin{definition}
For any instance $\nices=\instanceS$ of \hc{} put
$$G_\nices={K_\nices}^r\cup E_\nices$$
where $E_\nices$ is the set of edges from $A$ and $B$ to each of $X$, $x_i$, $S_j$, $T_{i,j}^{(l)}$, $P_i^{(l)}$, and $S_j^{(1)}$ for all possible $i,j,l$.
\end{definition}

Observe that $G_\nices$ is defined independently of any particular color assignment. Moreover:

\begin{lemma}
\label{lemma:dupsko1}
For any 2-colored instance $\nices$ we have $G_\nices={H_\nices}^r$.
\end{lemma}
\begin{proof}
Since $K_\nices\subset H_\nices$ then ${K_\nices}^r\subset {H_\nices}^r$. Now let us prove that $E_\nices\subset{H_\nices}^r$. In $H_\nices$ both $A$ and $B$ are in distance $2k=r-1$ from each of $S_j$ (because we had a proper 2-coloring), hence in ${H_\nices}^r$ we have edges from $A$ and $B$ to $S_j$ and $S_j^{(1)}$ (and to no further $S_j^{(l)}$). Both $A$ and $B$ are in distance $k+1$ from $X$ and $X$ is at most $k$ steps from each of $x_i$, $T_{i,j}^{(l)}$ and $P_i^{(l)}$, therefore $A$ and $B$ are at most $(k+1)+k=r$ steps from all these vertices. This proves that $G_\nices\subset {H_\nices}^r$.

Next we prove that ${H_\nices}^r\subset G_\nices$. Indeed, we have already checked that the edges between $\{A,B\}$ and the rest of the graph ${H_\nices}^r$ are exactly as described by $E_\nices$. Moreover, $A$ and $B$ are not adjacent in ${H_\nices}^r$ since $\dist_{H_\nices}(A,B)=2(k+1)=r+1$.

We only need to check that if two vertices $u$, $v$ other than $A$, $B$ are in distance at most $r$ in $H_\nices$ then they are in distance at most $r$ in $K_\nices$. If the shortest path from $u$ to $v$ in $H_\nices$ does not pass through $A$ or $B$ then of course it is true. If it does pass (say through $A$)
then both $u$ and $v$ must be from the set $\{x_i,S_j,X,T_{i,j}^{(l)},P_i^{(l)}\}$, because they cannot be further than $r-1=2k$ from $A$. The case $(u,v)=(S_j,S_{j'})$ is impossible, therefore at least one of $u$, $v$ (say $u$) is not any of the $S_j$. But then $\dist_{K_\nices}(u,X)\leq k$ and $\dist_{K_\nices}(v,X)\leq k+1$, so $\dist_{K_\nices}(u,v)\leq k+(k+1)=r$ as required.
\end{proof}

\begin{proof}[Proof of Theorem \ref{theorem:girth} for odd $r$]
Given an instance $\nices=\instanceS$ construct the graph $G_\nices$. If $\nices$ has a 2-coloring, then $G_\nices$ is the $r$-th power of a graph with girth at least $r+1$, namely $G_\nices={H_\nices}^r$ by Lemma \ref{lemma:dupsko1}.

For the inverse implication suppose that $G_\nices=H^r$ for some graph $H$. Define the coloring as follows: $x_i$ has color $A$ (resp. $B$) if there is a path of length at most $k$ from $x_i$ to $A$ (resp. $B$) in $H$. Clearly each $x_i$ is assigned at most one color since otherwise $A$ and $B$ would be adjacent in $H^r$.

The tail structure $S_j,S_j^{(1)},\ldots,S_j^{(r)}$ of each $S_j$ satisfies the assumptions of Lemma \ref{lemma:gadget}, so it enforces that in $H$:
\begin{itemize}
\item for every $j$ the $k$-neighbourhood of $S_j$ is precisely $\{x_i: x_i\in S_j\}\cup\{S_j^{(k)}\}$ (as in $K_\nices$),
\item $A$ and $B$ are exactly in distance $2k$ from each $S_j$ (by the definition of $E_\nices$).
\end{itemize}
Therefore for each $j$ there has to be at least one vertex in $\{x_i: x_i\in S_j\}$ that is $k$ steps from $A$ and at least one that is $k$ steps from $B$. This proves that the obtained coloring solves the \hc{} instance.
\end{proof}

\subsection{Case of even $r=2k$}
The argument in this case is similar, so we just outline the necessary changes. This is an extension of the construction from \cite{4auth}. 

Define the vertex set $V_\nices$ as
\begin{itemize}
\item $x_i$, $S_j$, $X$, $A$, $A'$, $B$, $B'$,
\item $T_{i,j}^{(l)}$ whenever $x_i\in S_j$ and $l=1,\ldots,k-1$,
\item $P_{i,A}^{(l)}$, $P_{i,B}^{(l)}$ for all $i$ and $l=1,\ldots,k-1$,
\item $S_j^{(l)}$ for all $j$ and $l=1,\ldots,r$.
\end{itemize}

The graph $K_\nices$ is built as previously, except that instead of a path $x_i\edge\ldots\edge P_i^{(k-1)}$ we have two paths $x_i\edge P_{i,A}^{(1)}\edge\ldots\edge P_{i,A}^{(k-1)}$ and $x_i\edge P_{i,B}^{(1)}\edge\ldots\edge P_{i,B}^{(k-1)}$. Additionally we provide $K_\nices$ with edges $A\edge A'$ and $B\edge B'$. See Fig.\ref{fig:red-even}.

For a 2-colored instance $\nices=\instanceS$ the graph $H_\nices$ is defined as the extension of $K_\nices$ by the edges:
\begin{itemize}
\item $P_{i,A}^{(k-1)}\edge A$ and $P_{i,B}^{(k-1)}\edge B'$ if $x_i$ has color $A$,
\item $P_{i,A}^{(k-1)}\edge A'$ and $P_{i,B}^{(k-1)}\edge B$ if $x_i$ has color $B$.
\end{itemize}

Note that $H_\nices$ has girth $2k+2=r+2$.

Eventually we set $G_\nices={K_\nices}^r\cup E_\nices$, where the edges of $E_\nices$ are from $A$, $A'$, $B$, $B'$ to each of $x_i$, $S_j$, $T_{i,j}^{(l)}$, $P_{i,A}^{(l)}$, $P_{i,B}^{(l)}$, $X$ and two extra edges $A\edge B'$ and $B\edge A'$.

Intuitively, $A$, $B$ encode the colors while $A'$, $B'$ encode the ``non-colors'', i.e. an element $x_i$ of color $A$ is connected to $A$ and $B'$ (``non-$B$''). This complication is necessary to ensure that $T_{i,j}^{(1)}$ and $A$ are connected by a path of length at most $r$ regardless of the color of $x_i$. (Indeed, had we only retained the vertices $A$, $B$ as before, then $\dist_{H_\nices}(T_{i,j}^{(1)}, A)$ would be either $2k-1=r-1$ if $x_i$ had color $A$ or $2k+1=r+1$ otherwise. Remember that the $r$-th power of $H_\nices$ must not depend on the color assignment.)

\begin{figure}[h!]
\includegraphics[scale=0.9]{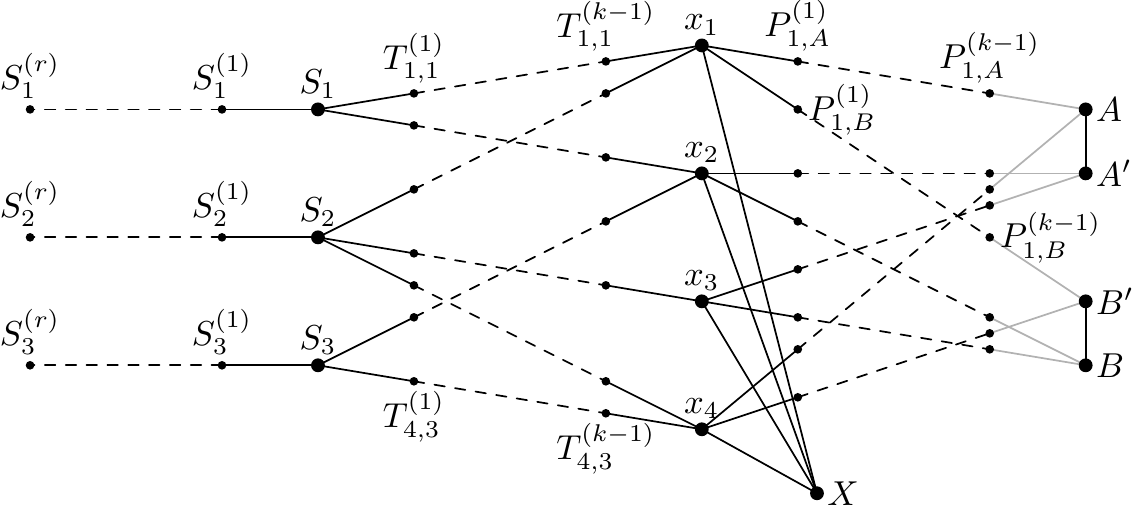}
\caption{The graphs $K_\nices$ and $H_\nices$ for even $r$. For description see Fig.\ref{fig:red-odd}.}
\label{fig:red-even}
\end{figure}

Also, when proving $G_\nices={H_\nices}^r$ one must consider the vertex pairs $(P_{i,A}^{(k-1)},S_j)$ and $(P_{i,B}^{(k-1)},S_j)$ separately, since whether they form an edge in ${K_\nices}^r$ or not depends on whether $x_i\in S_j$. However, this dependence does not change when passing from ${K_\nices}^r$ to ${H_\nices}^r$. 

With these alterations the proof in this case is analogous to the one for odd $r$.

\section{Conclusions and open problems}
\label{section:conclusions}

In this work we presented an efficient algorithmic solution to Levenshtein's reconstruction conjecture and we applied it to a more general, unrestricted $r$-th root problem. From a high-level perspective, it was possible because we could extract the ``core of the problem'' which has very few solutions (as the conjecture suggests), so we could hope that these can be found quickly. We also hope that the reverse flow of ideas is possible, so that some improved algorithmic edge-by-edge reconstruction technique might help resolve Levenshtein's conjecture.

Another (probably challenging) problem is to find a complete girth-parametrized complexity dichotomy, that is to close the gap between $r+1$ (or $r+2$) and $2r+3$. We believe that the $r$-th power recognition remains NP-complete even for graphs of girth $2r$.

In fact it would even be very interesting to find the complexity of SQUARE-OF-$\girth{\geq 5}$-GRAPH, to complete the complexity dichotomy of \cite{4auth} for $r=2$, because of possible implications in algebraic graph theory. Note, that if any algorithm for this problem is run with the complete graph $G=K_n$ as input, then the answer is ''yes'' if and only if there exists a graph on $n$ vertices that has girth at least $5$ and diameter $2$. By the Hoffman-Singleton theorem (see \cite{Singl,Biggs}) such a graph may exist only for $n=5,10,50$ and $3250$. The first three of these graphs are known, and the existence of the last one (for $n=3250$) is a long-standing open problem. Therefore, any efficient algorithm for SQUARE-OF-$\girth{\geq 5}$-GRAPH might (at least in principle) solve this problem.



\begin{thebibliography}{1}

\bibitem{Biggs} N.Biggs, Algebraic Graph Theory, Cambridge Univ. Press

\bibitem{CKL} Maw-Shang Chang, Ming-Tat Ko, Hsueh-I Lu, \textit{Linear-Time Algorithms for Tree Root Problems}, Proc. 10th SWAT, LNCS 4059 (2006)

\bibitem{EMR} F.Escalante, L.Montejano, T.Rojano, \textit{Characterization of $n$-path graphs and of graphs having $n$th root}, Journal of Combinatorial Theory, Series B, 16: 282-298 (1974)

\bibitem{4auth} Babak Farzad, Lap Chi Lau, Van Bang Le, Nguyen Ngoc Tuy, \textit{Computing Graph Roots Without Short Cycles}, Proc. 26th STACS (2009) 397-408

\bibitem{Gar} M.R.Garey, D.S.Johnson, \textit{Computers and Intractability --- A Guide to the Theory of NP-Completeness}, Freeman, Oxford, UK, 1979

\bibitem{KeaCor} P.E.Kearney, D.G.Corneil \textit{Tree powers}, Journal of Algorithms 29 (1998) 111-131

\bibitem{Kutz} Martin Kutz, \textit{The complexity of Boolean matrix root computation}, Theor. Comp. Sci. 325 (2004) 373-390

\bibitem{Lau} Lap Chi Lau, \textit{Bipartite Roots of Graphs}, ACM Transactions on Algorithms, Vol.2, No.2, April 2006, 178-208

\bibitem{LauCor} Lap Chi Lau, Derek G. Corneil \textit{Recognizing Powers of Proper Interval, Split and Chordal Graphs}, SIAM J. Discrete Math., Vol.18, No.1, 2004, 83-102

\bibitem{LeNg} Van Bang Le and Ngoc Tuy Nguyen, \textit{Hardness Results and Efficient Algorithms for Graph Powers}, WG 2009

\bibitem{Lev} V.I. Levenshtein, \textit{A conjecture on the reconstruction of graphs from metric balls of their vertices}, Discrete Mathematics 308(5-6): 993-998 (2008)

\bibitem{Lev+} V.I. Levenshtein, E.V. Konstantinova, E.Konstantinov, S.Molodtsov, \textit{Reconstruction of a graph from 2-vicinities of its vertices}, Discrete Applied Mathematics 156(9): 1399-1406 (2008)

\bibitem{Lin} Y.-L.Lin, S.S.Skiena, \textit{Algorithms for square roots of graphs}, SIAM J. Discrete Math. 8 (1995), 99-118

\bibitem{MotSud} R.Motwani, M.Sudan, \textit{Computing Roots of Graphs is Hard}, Discrete Applied Mathematics 54(1): 81-88 (1994)

\bibitem{Muk} A.Mukhopadhyay, \textit{The square root of a graph}, Journal of Combinatorial Theory, Series B, 2: 290-295 (1967)


\bibitem{Singl} R.R.Singleton, \textit{There is no irregular Moore graph}, American Mathematical Monthly 75, vol 1 (1968) 42-43


\end{thebibliography}
\end{document}